%% file: fixcofe.tex
\documentclass[acmsmall,nonacm]{acmart}

\usepackage[utf8]{inputenc}
\usepackage{cleveref}
\input{unicode-preamble.tex}
\newcommand{\n}[1]{\mathsf{#1}}

\newtheorem{theorem}{Theorem}

\newtheorem{lemma}{Lemma}

\begin{document}
\title{A fixed point theorem for COFEs}
\author{Stephen Dolan}
\email{stedolan@stedolan.net}
\maketitle

\section{Introduction}

A COFE (\emph{complete ordered family of equivalences}) is a set
equipped with an approximate equality $(≡_n)$ (where $a ≡_n b$ is read
as ``$a$ and $b$ are equal for $n$ steps'') and the ability to take
certain limits.  These were introduced by Di Gianantonio and
Miculan~\cite{cofes} to model recursive definitions, and have become a
standard tool in step-indexed logics such as Iris~\cite{iris}.

The main theorem that makes COFEs useful for modelling recursive
definitions is the following fixed-point theorem:  If $f$ is a
function from a COFE to itself which is \emph{contractive} in that:
\[
a \,≡_n\, b \quad⟶\quad f(a) \,≡_{n+1}\, f(b)
\]
then $f$ has a unique fixed point.

The purpose of this note is to introduce and prove a stronger
fixed-point principle for COFEs: to find a unique fixed point of $f$,
it is sufficient that $f$ be \emph{contractive on fixed points}:
\[
a \,≡_n\, b \,≡_n\, f(a) \,≡_n\, f(b) \quad⟶\quad f(a) \,≡_{n+1}\, f(b)
\]
That is, when showing that $f(a) ≡_{n+1} f(b)$, the user of this
fixed-point principle may assume not only that $a ≡_n b$, but also
that $a$ and $b$ are already \emph{partial fixed points}:
\[
a \,≡_n\, f(a) \qquad\qquad \,b ≡_n\, f(b)
\]
This stronger principle, based on a result by Fisher~\cite{fisher},
allows unique fixed points to be found for complex nested recursive
definitions that cannot be handled by contractivity alone.

\section{Background}

An \emph{ordered family of equivalences} (OFE) consists of a set $A$ and a sequence of equivalence relations $≡_n$ ($n ∈ ℕ$) on $A$, where:
\begin{itemize}
\item $≡₀$ is total: $∀a,\, b.\; a ≡₀ b$
\item The equivalence relations $≡_n$ are finer for larger $n$:
\[
(≡₀) ⊇ (≡₁) ⊇ (≡₂) ⊇ \dots
\]
\item Their intersection is equality: $(∀n.\; a ≡_n b) ⟶ a = b$
\end{itemize}

A sequence of elements $x_n$ is \emph{coherent} if $x_n ≡_n x_{n+1}$
for all $n$, and it has \emph{limit} $a$ if $x_n ≡_n a$ for all
$n$. Limits are unique when they exist, and a \emph{complete} OFE (or
COFE) is an OFE in which all coherent sequences have a (necessarily
unique) limit.

\paragraph*{OFEs as metric spaces}

An OFE can be viewed as a certain kind of metric space, by defining
the distance function $d$ as follows:
\[
d(a, b) = \inf \{ 2^{-n} \mid n ∈ ℕ,\; a ≡_n b \}
\]
The resulting metric is a \emph{1-bounded bisected ultrametric}:
\begin{itemize}
\item 1-bounded, because the maximum value of $d$ is 1 since $a ≡₀ b$ always
\item bisected, because $d$ takes only values $2^{-n}$ and $0$
\item ultrametric, because it obeys a stronger version of the triangle inequality, with $+$ replaced with $\n{max}$
\end{itemize}

The reader familiar with metric spaces might have expected a different
definition of completeness, in terms of \emph{Cauchy sequences}. In
OFE notation, a sequence $x_n$ is \emph{Cauchy} if its elements are
eventually approximately equal:
\[
∀n\, ∃k\, ∀i ≥ k,\, j ≥ k.\; x_i ≡_n x_j 
\]
and it converges to $a$ if its elements are eventually approximately equal to $a$:
\[
∀n\, ∃k\, ∀i ≥ k.\; x_i ≡_n a
\]
Some authors define COFEs as OFEs in which all Cauchy sequences
converge. This definition is equivalent to the simpler definition
above, because all coherent sequences are Cauchy and all Cauchy
sequences contain a coherent subsequence, and their limits agree.


\paragraph*{Contractive functions}

A function $f : A → A$ on a COFE $A$ is said to be \emph{contractive}
if:
\[
a ≡_n b ⟶ f(a) ≡_{n+1} f(b)
\]
The fixed point principle for COFEs introduced by Di Gianantonio and
Miculan~\cite{cofes} states that such functions have unique fixed
points, and moreover that this fixed point can be obtained by
iteration from an arbitrary starting point.
From the point of view of COFEs as metric spaces, this principle is a
restatement of Banach's fixed point theorem (from which we get the
term ``contractive'').

\begin{theorem}[Di Gianantonio and Miculan, Banach] If $f$ is contractive on an inhabited COFE, then $f$ has a unique fixed point, which is the limit of the iterates $f^n(x)$ starting from an arbitrary point $x$.
\end{theorem}

\section{A fixed point theorem}\label{sec:thm}

Say that a function $f : A → A$ is \emph{contractive on fixed points} if
\[
a ≡_n b ≡_n f(a) ≡_n f(b) ⟶ f(a) ≡_{n+1} f(b)
\]

All contractive functions are contractive on fixed points, but the
converse is not true (for an example, see \cref{sec:example}).
However, the fixed-point theorem still holds for functions merely
contractive on fixed points:

\begin{theorem}[cf. Fisher~\cite{fisher}]\label{theorem} If $f$ is contractive on fixed points on an inhabited COFE, then $f$ has a unique fixed point, which is the limit of the iterates $f^n(x)$ starting from an arbitrary point $x$.
\end{theorem}
\begin{proof}
First, show by induction on $n$ that, for arbitrary $x$:
\[
x ≡_n f(x) ⟶ f(x) ≡_{n+1} f(f(x))
\]
We use the c.f.p. property of $f$ and hence need that $x ≡_n f(x) ≡_n f(x) ≡_n f(f(x))$.
We already have $x ≡_n f(x)$, so need only show $f(x) ≡_n f(f(x))$.
When $n=0$, this holds by totality of $≡₀$.
When $n=n'+1$, we have $x ≡_{n'} f(x)$ and hence $f(x) ≡_{n'+1} f(f(x))$ by IH.

This property suffices to show that for arbitrary $x$, the sequence
$f^n(x)$ is coherent, so by completeness has limit $a$. We show that
it is a fixed point by proving $f(a) ≡_n a$ by induction on $n$, where
the $n=0$ case is trivial. For $n = n'+1$, by the IH and the convergence of $f^n(x)$ to $a$, we have:
\[
a ≡_{n'} f^{n'}(x) ≡_{n'} f(a) ≡_{n'} f^{n'+1}(x)
\]
and so $f(a) ≡_n f^{n}(x) ≡_n a$.

Finally, for uniqueness of the fixed point, we suppose that $f$ has
two fixed points $p, q$. We show that $p ≡_n q$ by induction on $n$,
relying on the c.f.p. property of $f$ to show that $p = f(p) ≡_{n'+1}
f(q) = q$ since $p ≡_{n'} q ≡_{n'} f(p) ≡_{n'} f(q)$.
\end{proof}

\section{Example}\label{sec:example}

Consider the following function, borrowed from Krstić and Matthews~\cite{krstic2003inductive}:
\[
f(x) \;=\; \text{if $x = 0$ then $0$ else $f(f(x-1))$}
\]
This function is total, always returns zero, and only ever recurses on
strictly smaller arguments. Yet it is not syntactically
structurally recursive, as it recurses on the argument $f(x-1)$.

We can view this definition as being the fixed point of an operator $T$:
\begin{align*}
T &: (ℕ → ℕ) → (ℕ → ℕ)\\ 
T f &= λx.\; \text{if $x = 0$ then $0$ else $f(f(x-1))$}
\end{align*}

We give the set $ℕ → ℕ$ the structure of a COFE by saying that
functions are approximately equal when they agree on a prefix:
\[
f ≡_n g \quad\text{iff}\quad ∀k < n.\; f(k) = g(k)
\]

In this COFE, the operator $T$ is not contractive, so Banach's fixed
point theorem cannot be used to prove it has a unique fixed point.
However, we can show that it is contractive on fixed points, by first
proving the following lemma about partial fixed points of $T$:

\begin{lemma}
If $T(f) ≡_n f$, then $∀k < n.\; f(k) = 0$
\end{lemma}
\begin{proof}Induction on $n$.\end{proof}

The operator $T$ is then shown to be contractive on fixed points,
since any partial fixed points of $T$ must satisfy the condition of
the above lemma. By a further application of the above lemma, this
unique fixed point is equal to $λx.\; 0$.

\section{Notes}

\Cref{theorem} is derived from a result of Fisher~\cite{fisher}, which
states that a function $f$ on a compact metric space has a unique fixed point if:
\[
d(f(x), f(y)) < \frac{1}{2}(d(x,f(y)) + d(f(x), y))
\]
\Cref{theorem} is not directly implied by this result, but is a sort
of ultrametric version of it, and can be proved by going through the
steps of Fisher's proof. However, the actual proof in \cref{sec:thm}
differs slightly from Fisher's, both in order to be constructive (Fisher's
proof does not directly show that the fixed point can be found by
iterating $f$), and because the bisected nature of COFEs removes the
need for compactness.

The theorem and the example of \cref{sec:example} have been formalised
in Coq, and the development is available from:

\begin{center}
\href{https://github.com/stedolan/cofe-fixpoints}{\tt https://github.com/stedolan/cofe-fixpoints}
\end{center}

\bibliographystyle{ACM-Reference-Format}
\bibliography{references}

\end{document}

%% file: unicode-preamble.tex
\DeclareUnicodeCharacter{2200}{\ensuremath{\forall}} 
\DeclareUnicodeCharacter{2083}{\ensuremath{_3}} 
\DeclareUnicodeCharacter{2102}{\ensuremath{\mathbb{C}}} 
\DeclareUnicodeCharacter{D7}{\ensuremath{\times}} 
\DeclareUnicodeCharacter{2287}{\ensuremath{\supseteq}} 
\DeclareUnicodeCharacter{2208}{\ensuremath{\in}} 
\DeclareUnicodeCharacter{2209}{\ensuremath{\not\in}} 
\DeclareUnicodeCharacter{230A}{\ensuremath{\lfloor{}}} 
\DeclareUnicodeCharacter{391}{\ensuremath{\Alpha}} 
\DeclareUnicodeCharacter{2293}{\ensuremath{\sqcap}} 
\DeclareUnicodeCharacter{1D502}{\ensuremath{\mathcal{y}}} 
\DeclareUnicodeCharacter{395}{\ensuremath{\Epsilon}} 
\DeclareUnicodeCharacter{2297}{\ensuremath{\otimes}} 
\DeclareUnicodeCharacter{399}{\ensuremath{\Iota}} 
\DeclareUnicodeCharacter{2218}{\ensuremath{\circ}} 
\DeclareUnicodeCharacter{229B}{\ensuremath{?}} 
\DeclareUnicodeCharacter{211A}{\ensuremath{\mathbb{Q}}} 
\DeclareUnicodeCharacter{39D}{\ensuremath{\Nu}} 
\DeclareUnicodeCharacter{3A1}{\ensuremath{\Rho}} 
\DeclareUnicodeCharacter{3A5}{\ensuremath{\Upsilon}} 
\DeclareUnicodeCharacter{22A7}{\ensuremath{\models}} 
\DeclareUnicodeCharacter{3A9}{\ensuremath{\Omega}} 
\DeclareUnicodeCharacter{2228}{\ensuremath{\lor}} 
\DeclareUnicodeCharacter{3B1}{\ensuremath{\alpha}} 
\DeclareUnicodeCharacter{B3}{\ensuremath{^3}} 
\DeclareUnicodeCharacter{3B5}{\ensuremath{\epsilon}} 
\DeclareUnicodeCharacter{3B9}{\ensuremath{\iota}} 
\DeclareUnicodeCharacter{3BD}{\ensuremath{\nu}} 
\DeclareUnicodeCharacter{1D53E}{\ensuremath{\mathbb{G}}} 
\DeclareUnicodeCharacter{3C1}{\ensuremath{\rho}} 
\DeclareUnicodeCharacter{22C3}{\ensuremath{\bigcup}} 
\DeclareUnicodeCharacter{1D542}{\ensuremath{\mathbb{K}}} 
\DeclareUnicodeCharacter{3C5}{\ensuremath{\upsilon}} 
\DeclareUnicodeCharacter{1D546}{\ensuremath{\mathbb{O}}} 
\DeclareUnicodeCharacter{3C9}{\ensuremath{\omega}} 
\DeclareUnicodeCharacter{1D54A}{\ensuremath{\mathbb{S}}} 
\DeclareUnicodeCharacter{1D54E}{\ensuremath{\mathbb{W}}} 
\DeclareUnicodeCharacter{1D4D3}{\ensuremath{\mathcal{D}}} 
\DeclareUnicodeCharacter{1D552}{\ensuremath{\mathbb{a}}} 
\DeclareUnicodeCharacter{1D4D7}{\ensuremath{\mathcal{H}}} 
\DeclareUnicodeCharacter{1D556}{\ensuremath{\mathbb{e}}} 
\DeclareUnicodeCharacter{1D4DB}{\ensuremath{\mathcal{L}}} 
\DeclareUnicodeCharacter{1D55A}{\ensuremath{\mathbb{i}}} 
\DeclareUnicodeCharacter{1D4DF}{\ensuremath{\mathcal{P}}} 
\DeclareUnicodeCharacter{1D55E}{\ensuremath{\mathbb{m}}} 
\DeclareUnicodeCharacter{1D4E3}{\ensuremath{\mathcal{T}}} 
\DeclareUnicodeCharacter{1D562}{\ensuremath{\mathbb{q}}} 
\DeclareUnicodeCharacter{2264}{\ensuremath{\leq}} 
\DeclareUnicodeCharacter{1D4E7}{\ensuremath{\mathcal{X}}} 
\DeclareUnicodeCharacter{1D566}{\ensuremath{\mathbb{u}}} 
\DeclareUnicodeCharacter{1D4EB}{\ensuremath{\mathcal{b}}} 
\DeclareUnicodeCharacter{1D56A}{\ensuremath{\mathbb{y}}} 
\DeclareUnicodeCharacter{1D4EF}{\ensuremath{\mathcal{f}}} 
\DeclareUnicodeCharacter{2070}{\ensuremath{^0}} 
\DeclareUnicodeCharacter{1D4F3}{\ensuremath{\mathcal{j}}} 
\DeclareUnicodeCharacter{2074}{\ensuremath{^4}} 
\DeclareUnicodeCharacter{1D4F7}{\ensuremath{\mathcal{n}}} 
\DeclareUnicodeCharacter{27F9}{\ensuremath{\Longrightarrow}} 
\DeclareUnicodeCharacter{1D4FB}{\ensuremath{\mathcal{r}}} 
\DeclareUnicodeCharacter{1D4FF}{\ensuremath{\mathcal{v}}} 
\DeclareUnicodeCharacter{1D501}{\ensuremath{\mathcal{x}}} 
\DeclareUnicodeCharacter{2080}{\ensuremath{_0}} 
\DeclareUnicodeCharacter{2203}{\ensuremath{\exists}} 
\DeclareUnicodeCharacter{2084}{\ensuremath{_4}} 
\DeclareUnicodeCharacter{2309}{\ensuremath{\rceil{}}} 
\DeclareUnicodeCharacter{210D}{\ensuremath{\mathbb{H}}} 
\DeclareUnicodeCharacter{392}{\ensuremath{\Beta}} 
\DeclareUnicodeCharacter{2115}{\ensuremath{\mathbb{N}}} 
\DeclareUnicodeCharacter{2294}{\ensuremath{\sqcup}} 
\DeclareUnicodeCharacter{396}{\ensuremath{\Zeta}} 
\DeclareUnicodeCharacter{2119}{\ensuremath{\mathbb{P}}} 
\DeclareUnicodeCharacter{39A}{\ensuremath{\Kappa}} 
\DeclareUnicodeCharacter{211D}{\ensuremath{\mathbb{R}}} 
\DeclareUnicodeCharacter{39E}{\ensuremath{\Xi}} 
\DeclareUnicodeCharacter{22A4}{\ensuremath{\top}} 
\DeclareUnicodeCharacter{2227}{\ensuremath{\land}} 
\DeclareUnicodeCharacter{3A6}{\ensuremath{\Phi}} 
\DeclareUnicodeCharacter{3B2}{\ensuremath{\beta}} 
\DeclareUnicodeCharacter{3B6}{\ensuremath{\zeta}} 
\DeclareUnicodeCharacter{1D539}{\ensuremath{\mathbb{B}}} 
\DeclareUnicodeCharacter{3BA}{\ensuremath{\kappa}} 
\DeclareUnicodeCharacter{1D53D}{\ensuremath{\mathbb{F}}} 
\DeclareUnicodeCharacter{3BE}{\ensuremath{\xi}} 
\DeclareUnicodeCharacter{1D541}{\ensuremath{\mathbb{J}}} 
\DeclareUnicodeCharacter{22C0}{\ensuremath{\bigland}} 
\DeclareUnicodeCharacter{2243}{\ensuremath{\simeq}} 
\DeclareUnicodeCharacter{3C6}{\ensuremath{\phi}} 
\DeclareUnicodeCharacter{1D54D}{\ensuremath{\mathbb{V}}} 
\DeclareUnicodeCharacter{1D4D0}{\ensuremath{\mathcal{A}}} 
\DeclareUnicodeCharacter{21D2}{\ensuremath{\Rightarrow}} 
\DeclareUnicodeCharacter{1D555}{\ensuremath{\mathbb{d}}} 
\DeclareUnicodeCharacter{1D4D4}{\ensuremath{\mathcal{E}}} 
\DeclareUnicodeCharacter{1D559}{\ensuremath{\mathbb{h}}} 
\DeclareUnicodeCharacter{1D4D8}{\ensuremath{\mathcal{I}}} 
\DeclareUnicodeCharacter{1D55D}{\ensuremath{\mathbb{l}}} 
\DeclareUnicodeCharacter{1D4DC}{\ensuremath{\mathcal{M}}} 
\DeclareUnicodeCharacter{1D561}{\ensuremath{\mathbb{p}}} 
\DeclareUnicodeCharacter{1D4E0}{\ensuremath{\mathcal{Q}}} 
\DeclareUnicodeCharacter{1D565}{\ensuremath{\mathbb{t}}} 
\DeclareUnicodeCharacter{1D4E4}{\ensuremath{\mathcal{U}}} 
\DeclareUnicodeCharacter{27E6}{\ensuremath{\llbracket}} 
\DeclareUnicodeCharacter{1D569}{\ensuremath{\mathbb{x}}} 
\DeclareUnicodeCharacter{1D4E8}{\ensuremath{\mathcal{Y}}} 
\DeclareUnicodeCharacter{1D4EC}{\ensuremath{\mathcal{c}}} 
\DeclareUnicodeCharacter{1D4F0}{\ensuremath{\mathcal{g}}} 
\DeclareUnicodeCharacter{1D4F4}{\ensuremath{\mathcal{k}}} 
\DeclareUnicodeCharacter{27F6}{\ensuremath{\longrightarrow}} 
\DeclareUnicodeCharacter{1D4F8}{\ensuremath{\mathcal{o}}} 
\DeclareUnicodeCharacter{207B}{\ensuremath{^{-}}} 
\DeclareUnicodeCharacter{1D4FC}{\ensuremath{\mathcal{s}}} 
\DeclareUnicodeCharacter{2081}{\ensuremath{_1}} 
\DeclareUnicodeCharacter{1D500}{\ensuremath{\mathcal{w}}} 
\DeclareUnicodeCharacter{2202}{\ensuremath{\partial}} 
\DeclareUnicodeCharacter{2A06}{\ensuremath{\bigsqcup}} 
\DeclareUnicodeCharacter{2308}{\ensuremath{\lceil{}}} 
\DeclareUnicodeCharacter{2291}{\ensuremath{\sqsubseteq}} 
\DeclareUnicodeCharacter{393}{\ensuremath{\Gamma}} 
\DeclareUnicodeCharacter{2295}{\ensuremath{\oplus}} 
\DeclareUnicodeCharacter{397}{\ensuremath{\Eta}} 
\DeclareUnicodeCharacter{39B}{\ensuremath{\Lambda}} 
\DeclareUnicodeCharacter{39F}{\ensuremath{\Omicron}} 
\DeclareUnicodeCharacter{3A3}{\ensuremath{\Sigma}} 
\DeclareUnicodeCharacter{22A5}{\ensuremath{\bot}} 
\DeclareUnicodeCharacter{2124}{\ensuremath{\mathbb{Z}}} 
\DeclareUnicodeCharacter{3A7}{\ensuremath{\Chi}} 
\DeclareUnicodeCharacter{2026}{\ensuremath{\dots}} 
\DeclareUnicodeCharacter{1F4A9}{\ensuremath{\LaTeX}} 
\DeclareUnicodeCharacter{222A}{\ensuremath{\cup}} 
\DeclareUnicodeCharacter{3B3}{\ensuremath{\gamma}} 
\DeclareUnicodeCharacter{3B7}{\ensuremath{\eta}} 
\DeclareUnicodeCharacter{B9}{\ensuremath{^1}} 
\DeclareUnicodeCharacter{1D538}{\ensuremath{\mathbb{A}}} 
\DeclareUnicodeCharacter{3BB}{\ensuremath{\lambda}} 
\DeclareUnicodeCharacter{1D53C}{\ensuremath{\mathbb{E}}} 
\DeclareUnicodeCharacter{3BF}{\ensuremath{\omicron}} 
\DeclareUnicodeCharacter{22C1}{\ensuremath{\biglor}} 
\DeclareUnicodeCharacter{1D540}{\ensuremath{\mathbb{I}}} 
\DeclareUnicodeCharacter{3C3}{\ensuremath{\sigma}} 
\DeclareUnicodeCharacter{22C5}{\ensuremath{\cdot}} 
\DeclareUnicodeCharacter{1D544}{\ensuremath{\mathbb{M}}} 
\DeclareUnicodeCharacter{3C7}{\ensuremath{\chi}} 
\DeclareUnicodeCharacter{1D54C}{\ensuremath{\mathbb{U}}} 
\DeclareUnicodeCharacter{1D4D1}{\ensuremath{\mathcal{B}}} 
\DeclareUnicodeCharacter{1D550}{\ensuremath{\mathbb{Y}}} 
\DeclareUnicodeCharacter{1D4D5}{\ensuremath{\mathcal{F}}} 
\DeclareUnicodeCharacter{1D554}{\ensuremath{\mathbb{c}}} 
\DeclareUnicodeCharacter{1D4D9}{\ensuremath{\mathcal{J}}} 
\DeclareUnicodeCharacter{1D558}{\ensuremath{\mathbb{g}}} 
\DeclareUnicodeCharacter{1D4DD}{\ensuremath{\mathcal{N}}} 
\DeclareUnicodeCharacter{1D55C}{\ensuremath{\mathbb{k}}} 
\DeclareUnicodeCharacter{1D4E1}{\ensuremath{\mathcal{R}}} 
\DeclareUnicodeCharacter{1D560}{\ensuremath{\mathbb{o}}} 
\DeclareUnicodeCharacter{1D4E5}{\ensuremath{\mathcal{V}}} 
\DeclareUnicodeCharacter{1D564}{\ensuremath{\mathbb{s}}} 
\DeclareUnicodeCharacter{27E7}{\ensuremath{\rrbracket}} 
\DeclareUnicodeCharacter{1D4E9}{\ensuremath{\mathcal{Z}}} 
\DeclareUnicodeCharacter{1D568}{\ensuremath{\mathbb{w}}} 
\DeclareUnicodeCharacter{1D4ED}{\ensuremath{\mathcal{d}}} 
\DeclareUnicodeCharacter{1D4F1}{\ensuremath{\mathcal{h}}} 
\DeclareUnicodeCharacter{2192}{\ensuremath{\rightarrow}} 
\DeclareUnicodeCharacter{1D4F5}{\ensuremath{\mathcal{l}}} 
\DeclareUnicodeCharacter{1D4F9}{\ensuremath{\mathcal{p}}} 
\DeclareUnicodeCharacter{207A}{\ensuremath{^{+}}} 
\DeclareUnicodeCharacter{1D4FD}{\ensuremath{\mathcal{t}}} 
\DeclareUnicodeCharacter{1D503}{\ensuremath{\mathcal{z}}} 
\DeclareUnicodeCharacter{2082}{\ensuremath{_2}} 
\DeclareUnicodeCharacter{2A05}{\ensuremath{\bigsqcap}} 
\DeclareUnicodeCharacter{2286}{\ensuremath{\subseteq}} 
\DeclareUnicodeCharacter{230B}{\ensuremath{\rfloor{}}} 
\DeclareUnicodeCharacter{2292}{\ensuremath{\sqsupseteq}} 
\DeclareUnicodeCharacter{394}{\ensuremath{\Delta}} 
\DeclareUnicodeCharacter{398}{\ensuremath{\Theta}} 
\DeclareUnicodeCharacter{39C}{\ensuremath{\Mu}} 
\DeclareUnicodeCharacter{3A0}{\ensuremath{\Pi}} 
\DeclareUnicodeCharacter{22A2}{\ensuremath{\vdash}} 
\DeclareUnicodeCharacter{3A4}{\ensuremath{\Tau}} 
\DeclareUnicodeCharacter{2229}{\ensuremath{\cap}} 
\DeclareUnicodeCharacter{3A8}{\ensuremath{\Psi}} 
\DeclareUnicodeCharacter{B2}{\ensuremath{^2}} 
\DeclareUnicodeCharacter{3B4}{\ensuremath{\delta}} 
\DeclareUnicodeCharacter{3B8}{\ensuremath{\theta}} 
\DeclareUnicodeCharacter{1D53B}{\ensuremath{\mathbb{D}}} 
\DeclareUnicodeCharacter{3BC}{\ensuremath{\mu}} 
\DeclareUnicodeCharacter{3C0}{\ensuremath{\pi}} 
\DeclareUnicodeCharacter{1D543}{\ensuremath{\mathbb{L}}} 
\DeclareUnicodeCharacter{22C2}{\ensuremath{\bigcap}} 
\DeclareUnicodeCharacter{2245}{\ensuremath{\cong}} 
\DeclareUnicodeCharacter{3C4}{\ensuremath{\tau}} 
\DeclareUnicodeCharacter{3C8}{\ensuremath{\psi}} 
\DeclareUnicodeCharacter{1D54B}{\ensuremath{\mathbb{T}}} 
\DeclareUnicodeCharacter{1D54F}{\ensuremath{\mathbb{X}}} 
\DeclareUnicodeCharacter{1D553}{\ensuremath{\mathbb{b}}} 
\DeclareUnicodeCharacter{1D4D2}{\ensuremath{\mathcal{C}}} 
\DeclareUnicodeCharacter{1D557}{\ensuremath{\mathbb{f}}} 
\DeclareUnicodeCharacter{1D4D6}{\ensuremath{\mathcal{G}}} 
\DeclareUnicodeCharacter{1D55B}{\ensuremath{\mathbb{j}}} 
\DeclareUnicodeCharacter{1D4DA}{\ensuremath{\mathcal{K}}} 
\DeclareUnicodeCharacter{1D55F}{\ensuremath{\mathbb{n}}} 
\DeclareUnicodeCharacter{1D4DE}{\ensuremath{\mathcal{O}}} 
\DeclareUnicodeCharacter{2261}{\ensuremath{\equiv}} 
\DeclareUnicodeCharacter{1D563}{\ensuremath{\mathbb{r}}} 
\DeclareUnicodeCharacter{1D4E2}{\ensuremath{\mathcal{S}}} 
\DeclareUnicodeCharacter{2265}{\ensuremath{\geq}} 
\DeclareUnicodeCharacter{1D567}{\ensuremath{\mathbb{v}}} 
\DeclareUnicodeCharacter{1D4E6}{\ensuremath{\mathcal{W}}} 
\DeclareUnicodeCharacter{1D56B}{\ensuremath{\mathbb{z}}} 
\DeclareUnicodeCharacter{1D4EA}{\ensuremath{\mathcal{a}}} 
\DeclareUnicodeCharacter{1D4EE}{\ensuremath{\mathcal{e}}} 
\DeclareUnicodeCharacter{1D4F2}{\ensuremath{\mathcal{i}}} 
\DeclareUnicodeCharacter{1D4F6}{\ensuremath{\mathcal{m}}} 
\DeclareUnicodeCharacter{1D4FA}{\ensuremath{\mathcal{q}}} 
\DeclareUnicodeCharacter{1D4FE}{\ensuremath{\mathcal{u}}} 